\documentclass[journal]{IEEEtran}
\usepackage{cite}
\usepackage{subfig}
\usepackage{amsmath,amssymb,amsfonts,amsthm}
\usepackage{algorithmic}
\usepackage{graphicx}
\usepackage{hyperref}
\usepackage{multirow}
\usepackage{multicol}
\usepackage{booktabs}
\hypersetup{colorlinks=true,urlcolor=red}
\usepackage{bm} % bold
\usepackage{xspace}
\newcommand{\xmath}[1] {\ensuremath{#1}\xspace}
\newcommand{\blmath}[1] {\xmath{\bm{#1}}}

\newcommand{\Xb}{{\blmath X}}
\newcommand{\Yb}{{\blmath Y}}
\newcommand{\Zb}{{\blmath Z}}

\newcommand{\wb}{{\blmath w}}
\newcommand{\xb}{{\blmath x}}
\newcommand{\yb}{{\blmath y}}
\newcommand{\zb}{{\blmath z}}

\newcommand{\Hc}{\mathcal{H}}

\newcommand{\Xc}{\mathcal{X}}
\newcommand{\Yc}{\mathcal{Y}}

\newcommand{\Rd}{{\mathbb R}}

\newcommand{\Pc}{{{\mathcal P}}}

\newcommand{\Kd}{\mathbb{K}}

\newcommand{\beq}{\begin{equation}}
\newcommand{\eeq}{\end{equation}}
\newcommand{\beqa}{\begin{eqnarray}}
\newcommand{\eeqa}{\end{eqnarray}}
\newcommand{\Fc}{{\mathcal F}}

\newtheorem{theorem}{Theorem}
\newtheorem{proposition}[theorem]{Proposition}

\newcommand{\Mc}{\mathcal{M}}

\begin{document}

\title{Unpaired Deep Learning for Accelerated MRI using Optimal Transport Driven CycleGAN}
\date{\vspace{-4ex}}

\author{Gyutaek~Oh, 
        Byeongsu Sim, 
        HyungJin Chung, 
      {Leonard Sunwoo,}
        and~Jong~Chul~Ye,~\IEEEmembership{Fellow,~IEEE}% <-this % stops a space
\thanks{G. Oh, H. Chung, and J. C. Ye are with the Department of Bio and Brain Engineering, 
		Korea Advanced Institute of Science and Technology (KAIST), 
		Daejeon 34141, Republic of Korea (e-mail: \{okt0711, hj.chung, jong.ye\}@kaist.ac.kr). 
		B. Sim is with the Department of Mathematical Sciences, KAIST 
		(e-mail: byeongsu.s@kaist.ac.kr).
		L. Sunwoo is with the Department of Radiology, Seoul National University College of Medicine, Seoul National University Bundang Hospital, Seongnam 13620, Republic of Korea.
		J.C. Ye is also with the Department of Mathematical Sciences, KAIST.} 
}

\maketitle

\begin{abstract}
Recently, deep learning approaches for accelerated MRI have been extensively studied thanks
to their high performance reconstruction in spite of significantly  reduced run-time complexity. 
These neural networks are usually trained in a supervised manner, so 
matched pairs of subsampled and fully sampled $k$-space data are required. 
Unfortunately, it is often difficult to acquire matched fully sampled $k$-space data, since the acquisition of fully sampled
$k$-space data requires long scan time and often leads to the change of the acquisition protocol.
Therefore, unpaired deep learning without matched label data has become a very important research 
topic. 
In this paper, we propose an unpaired deep learning approach using a optimal transport driven cycle-consistent generative 
adversarial network (OT-cycleGAN) that employs a single pair of generator and  discriminator.
The proposed OT-cycleGAN architecture is rigorously derived from a dual formulation of the
optimal transport formulation using a specially designed penalized least squares cost.
The experimental results show that our method can reconstruct high resolution MR images from 
accelerated $k$-space data from both single and multiple coil acquisition, without requiring matched reference data.

\end{abstract}


\IEEEpeerreviewmaketitle

\section{Introduction}\label{sec:introduction}
\IEEEPARstart{M}{agnetic} resonance image (MRI) is a non-invasive medical imaging modality which provides high resolution anatomical and functional images without concerning radiation hazard. 
Unfortunately, due to the nature of the MR acquisition physics, MR scanning time is relatively  long, so many researchers have studied various methods to reduce the acquisition time.
In this regard, accelerated MRI, which relies on the sparse sampling of the $k$-space data, is one of the most important research efforts to reduce the scan time of MRI. 
Unfortunately, aliasing artifacts are unavoidable in accelerated MRI, because of the $k$-space undersampling.
 
For the last decade, compressed sensing (CS) \cite{CaRoTa06,donoho2006compressed} has been a  main research thrust for accelerated MRI \cite{lustig2008compressed, jin2016general}. 
Compressed sensing algorithms can reconstruct high resolution images from undersampled $k$-space data by exploiting the sparsity of a signal in some transform domains, e.g. wavelet transform, etc. 
Although CS methods show high performance, computational complexity is quite significant and hyper-parameter tuning is a difficult task in practice.

Recently, deep learning approaches have been successfully used for various tasks, such as image classification \cite{krizhevsky2012imagenet}, segmentation \cite{ronneberger2015u}, super-resolution problems \cite{dong2015image}, etc. 
Inspired by these successes, many researchers have investigated deep learning approaches for MR reconstruction problems and successfully demonstrated significant performance gain \cite{wang2016accelerating,han2017deep,hammernik2018learning,lee2018deep,schlemper2018deep,zhu2018image, schlemper2018stochastic, eo2018kiki,han2019k, wang2019dimension, liu2019ifr, wang2020deepcomplexmri, sriram2020end, mardani2017deep, yang2017dagan, quan2018compressed, seitzer2018adversarial, wang2019accelerated, narnhofer2019inverse}.
Most of the existing MR deep learning approaches require  fully sampled $k$-space data that can be retrospectively subsampled so that supervised neural network training is feasible.

When the matched fully sampled $k$-space data are not available, generative adversarial network (GAN) \cite{goodfellow2014generative} can be a useful approach thanks to its ability to generate a target data distribution even from a random distribution.
In particular, the mathematical structure has been extensively investigated for several algorithms; for example, Wasserstein GAN (WGAN) \cite{arjovsky2017wasserstein} was derived from the mathematical theory of optimal transport (OT) \cite{villani2008optimal, peyre2019computational}.
In optimal transport theory, one distribution can be transported to another distribution by minimizing the average transportation cost with respect to joint measures, so it provides an unsupervised way of learning the transportation map between two distributions. 
However, artificial features can be generated by these GAN approaches due to mode collapsing. 
Thus, cycle-consistent generative adversarial network (cycleGAN) \cite{zhu2017unpaired} has been investigated for unsupervised learning in the medical imaging field \cite{kang2019cycle, liang2019generating}.

In our recent paper \cite{sim2019optimal, 9136890}, we have revealed that a family of cycleGAN architecture can be derived from a dual formulation of the optimal transport by using a specially designed penalized least squares (PLS) cost  as a transportation cost. 
The theory provides a systematic approach to design a novel optimal transport driven cycleGAN (OT-cycleGAN) architecture for various inverse problems, and preliminary results for MR reconstruction from sparse Fourier samples was also provided  \cite{sim2019optimal}.
However, in applying this algorithm to real MR problems including parallel imaging, significant technical issues have been identified.
In particular, for parallel imaging, the direct application of the theory in \cite{sim2019optimal} leads to a complicated discriminator architecture that compares all channel-wise reconstruction, which makes the algorithm difficult to train.
To address this, here we introduce a novel compositional discriminator and cycle-consistency term using the square-root of sum-of-the squares (SSoS) operation, and provide a mathematical  theory to guarantee that it is still valid dual formulation from optimal transport problem. 
Accordingly, the main architectural novelty in  \cite{sim2019optimal}, which replaces one generator with a deterministic Fourier transform, is still retained.

Due to the simplified network architecture, the proposed OT-cycleGAN architecture not only improves the stability of network training but also reduces training time. 
Furthermore, experimental results show that the proposed method can outperform the conventional cycleGAN and achieve competitive performance to supervised learning approaches. 

This paper is structured as follows.
Section~\ref{sec:related works} reviews the existing deep learning approaches for MR, cycleGAN, and the optimal transport for inverse problems.
In Section \ref{sec:theory}, the theory for OT-cycleGAN design is explained, and then a novel cycleGAN architecture for accelerated MRI for both single and parallel imaging is derived from the theory. 
Section \ref{sec:method} explains the resulting network architecture, experimental datasets, and training details. 
The experimental results are given in Section \ref{sec:result}. 
Then Section \ref{sec:discussion} provides discussion, which is followed by conclusions in Section \ref{sec:conclusion}.

\section{Related works}\label{sec:related works}
\subsection{Deep Learning Methods for Accelerated MRI}
Many deep learning algorithms for accelerated MRI have been implemented in the image domain. 
Wang $et$ $al.$\cite{wang2016accelerating} used convolutional neural network (CNN) to reconstruct fully sampled MR images from downsampled MR images. 
In \cite{han2017deep}, transfer learning approach  was adopted to remove the streaking artifacts in radial acquisition.
Unfolded variational approach for compressed sensing MRI was implemented as image domain learning\cite{hammernik2018learning}, and deep cascaded CNN with the data consistency layer was successfully demonstrated for dynamic MRI\cite{schlemper2018deep, schlemper2018stochastic}.
Recently, Sriram $et$ $al.$\cite{sriram2020end} proposed end-to-end variational networks which estimate coil sensitivity maps and exploit them to reconstruct MR images.
As MR data are complex valued, the usual approach is to split real and imaginary values into separate channels.
Instead, Lee $et$ $al.$\cite{lee2018deep} proposed deep residual learning networks that consist of magnitude and phase networks, and Wang $et$ $al.$\cite{wang2020deepcomplexmri} suggested a neural network architecture with direct complex convolution.

Although image domain learnings are popular, they often suffer from the blurring artifacts if the training data are not sufficient.
To overcome this limitation, several algorithms are implemented in the $k$-space domain.
Akcakaya  $et$ $al.$\cite{akccakaya2019scan} directly estimate the GRAPPA interpolation kernel using a neural network without training data.
Han $et$ $al.$\cite{han2019k} demonstrated that CNN can learn $k$-space data interpolation kernel from training data, and verified that their method outperforms other image domain reconstruction methods. 
More sophisticated deep learning algorithms train the neural network in both of image and $k$-space domains for further improvement.
Kiki-net\cite{eo2018kiki} is composed of CNNs for $k$-space domain learning and image domain learning, respectively, along with data consistency layers. 
Wang $et$ $al.$\cite{wang2019dimension} also proposed DIMENSION network, which consists of a frequency prior network and a spatial prior network, for dynamic MR imaging.

Inspired by the success of generative adversarial network (GAN), Mardani $et$ $al.$\cite{mardani2017deep} implemented generative adversarial networks for CS-MRI. 
They adopted the adversarial loss to the cost function, so that the generator can recover detailed textures.
In \cite{yang2017dagan,seitzer2018adversarial}, the adversarial loss was combined with the perceptual loss for visual improvement of reconstructed images. 
Wang $et$ $al.$\cite{wang2019accelerated} proposed dual-domain GAN for accelerated MRI.
The generator and the discriminator in \cite{wang2019accelerated} were trained on both of image and $k$-space domains, so that fine details of MR images could be recovered.
Quan $et$ $al.$\cite{quan2018compressed} also introduced GAN with a cyclic loss for dynamic MRI.

One of the fundamental limitations of the aforementioned works is that they require matched pairs which contain fully sampled and downsampled $k$-spaces or images.
Although GAN could be used for unpaired training, most of the aforementioned algorithms adopted GAN for the refinement of visual contents, not for the reconstruction with unpaired data.
Narnhofer $et$ $al.$\cite{narnhofer2019inverse} tried to address this issue by proposing inverse GANs for accelerated MRI.
Specifically, they trained GANs to map a latent vector to complex-valued MR image first. 
Then, by presenting an optimization process called GAN prior adaptation, they optimized networks to reconstruct accelerated MR images with high frequency details. 
Although the work by Quan $et$ $al.$\cite{quan2018compressed} could be in principle used for unsupervised training, their main focus was using GAN for the visual refinement in supervised training.

\subsection{CycleGAN and Optimal Transport}
Wasserstein GAN (WGAN)\cite{arjovsky2017wasserstein} is an important landmark in GAN research, thanks to the discovery of an explicit mathematical link to the optimal transport (OT). 
Optimal transport is a mathematical tool to compare two measures in a Lagrangian framework \cite{villani2008optimal,peyre2019computational}.
Since OT provides a means to transport one probability measure to another by minimizing the average transportation cost {with respect to joint measures}, it provides a mathematical means for unsupervised learning between two probability distributions.

Inspired by the link between GAN and OT\cite{arjovsky2017wasserstein}, many researchers have provided analyses of GAN from the optimal transport point of view. 
For example, Genevay $et$ $al.$\cite{genevay2018learning} proposed an OT loss called Sinkhorn loss by introducing the entropic smoothing.
GAN with Sinkhorn loss showed stable training and reduced computational complexity.
Also, Salimans $et$ $al.$\cite{salimans2018improving} presented optimal transport GAN with a novel metric. 
They defined the mini-batch energy distance which is more discriminative than other existing distances.

To deal with the mode-collapsing in GAN-based unpaired image-to-image translation, Zhu $et$ $al.$ proposed cycle consistent generative adversarial network (cycleGAN)\cite{zhu2017unpaired}.
By extending cycleGAN, Almahairi $et$ $al.$\cite{almahairi2018augmented} proposed augmented cycleGAN for learning many-to-many mapping.
In \cite{yuan2018unsupervised}, cycleGAN was used for image super resolution.

\subsection{Main Contribution}
Although several studies have provided analyses of GAN from the optimal transport point of view, these are for the original GAN and its variants, and to the best of our knowledge, the relationship between cycleGAN and optimal transport has not been studied extensively.
In our companion paper\cite{sim2019optimal}, we derived a family of cycleGAN architecture from optimal transport perspective. 
By extending this idea for real applications, in this paper we provide a mathematical derivation of a novel OT-cycleGAN architecture for accelerated MRI. 
This paper is the first real world application of the theory for accelerated MRI, which can be used not only for single coil imaging but also for parallel imaging.

\begin{figure*}[!t] 	
	\centerline{\includegraphics[width=0.99\linewidth]{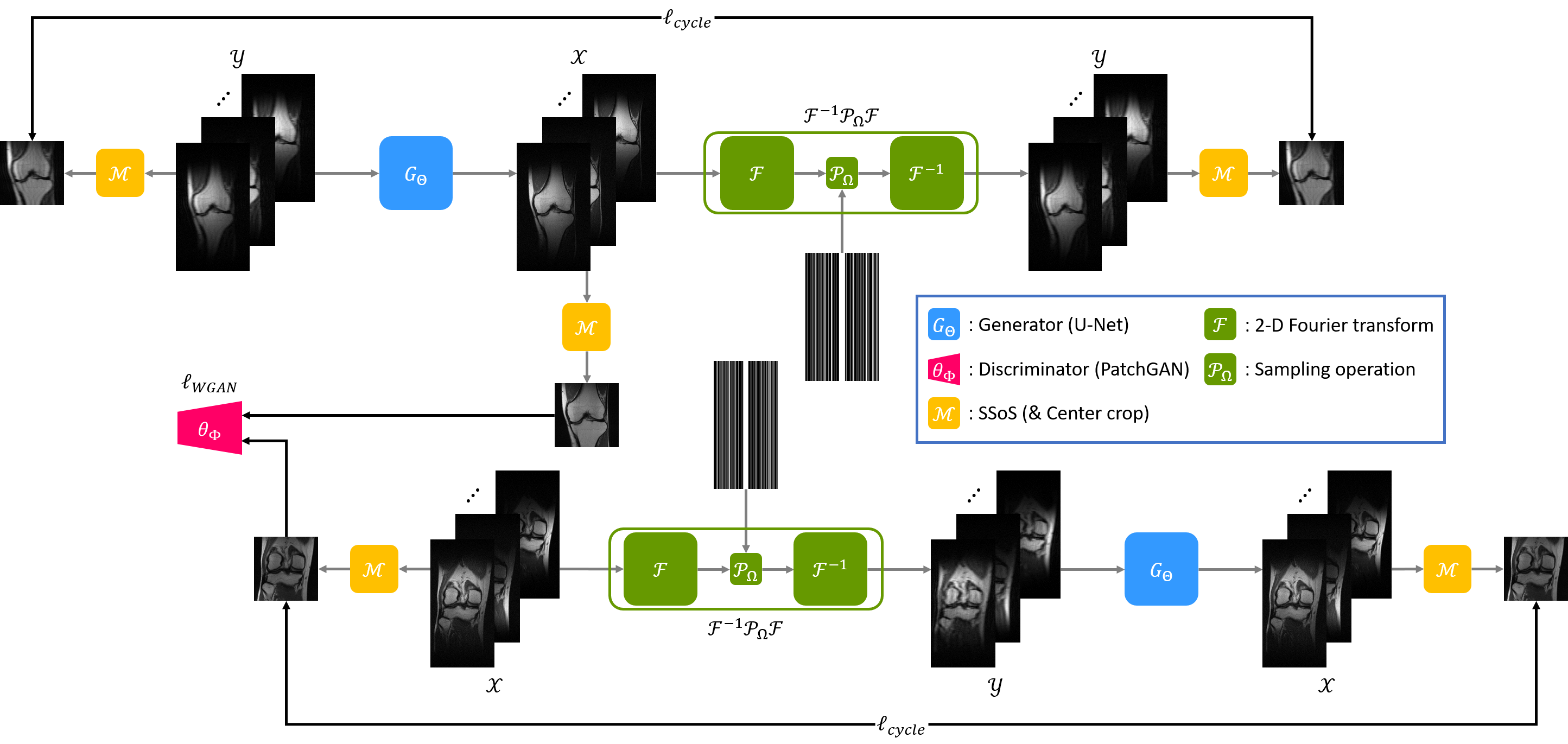}}
	\caption{Proposed cycleGAN architecture for accelerated MRI. The generator to produce 
		downsampled images from fully sampled images is replaced by a deterministic operation using 2-D Fourier transform and 
		$k$-space sampling operation. 
		The MR images in the top and bottom branches are different because they are unpaired data. The sampling mask in the bottom branch is randomly generated at every step.
		Note that the cycle-consistency loss and the WGAN loss are calculated using the sum-of-the squares images.}
	\vspace{-0.5cm}
	\label{fig:cycleGAN}
\end{figure*}

\section{Theory}\label{sec:theory}
\subsection{Optimal Transport Driven CycleGAN: A Review}
Here, we briefly review the OT-driven cycleGAN design in our companion paper\cite{sim2019optimal} to make the paper self-contained.

In inverse problems, a noisy measurement $\yb \in \Yc$ from an unobserved image $\xb \in \Xc$ is modeled by
\begin{eqnarray}
\yb&=&\Hc \xb +\wb \ ,
\end{eqnarray}
where $\wb$ is the measurement noise and $\Hc : \Xc \mapsto \Yc$ is the measurement operator.
In our companion paper\cite{sim2019optimal}, we propose a penalized least squares (PLS) cost function with a novel deep learning prior as follows:
\begin{eqnarray}\label{eq:cost0}
c(\xb,\yb;\Theta,\Hc)=\|\yb-\Hc\xb\|+ \| G_\Theta(\yb) -\xb\| \ ,
\end{eqnarray}
where $G_\Theta$ is a neural network with the network parameter $\Theta$ and the input $\yb$.
For unsupervised learning using optimal transport, we assume that \eqref{eq:cost0} is a transport cost between the two domains $\Xc$ and $\Yc$ with probability distributions with measures $\mu$ and $\nu$, respectively.
Then, we showed that Kantorovich dual optimal transport formulation\cite{villani2008optimal}, which minimizes the average transport cost for between two measures $\Xc$ and $\Yc$, is given by \cite{sim2019optimal}:
\begin{align}
\min_{\Theta, \Hc} \Kd(\Theta, \Hc) = \min_{\Theta, \Hc} \max_{\Phi, \Xi} 
\ell(\Theta, \Hc; \Phi, \Xi) \ ,
\end{align}
where
\begin{eqnarray}
\ell(\Theta, \Hc;\Phi, \Xi) = \gamma \ell_{cycle}(\Theta, \Hc) + \ell_{WGAN}(\Theta, \Hc;\Phi, \Xi) \ ,
\end{eqnarray}
where $\gamma$ is an appropriate hyperparameter, $\ell_{cycle}$ denotes the cycle-consistency loss, and $\ell_{WGAN}$ is the Wasserstein GAN\cite{arjovsky2017wasserstein} loss. 
More specifically, $\ell_{cycle}$ is given by \cite{sim2019optimal}:
\begin{eqnarray}
\begin{split}
\ell_{cycle}(\Theta, \Hc) =  \int_\Xc \|\xb - G_\Theta(\Hc \xb)\|d\mu(\xb) \\
+  \int_\Yc \|\yb - \Hc G_\Theta(\yb)\| d\nu(\yb) \ ,
\end{split}
\label{eq:cycle loss}
\end{eqnarray}
and  $\ell_{WGAN}$ is given by \cite{sim2019optimal}:
\begin{eqnarray}
\begin{split}
&\ell_{WGAN}(\Theta, \Hc;\Phi, \Xi) \\
&= \left(\int_\Xc \varphi_\Phi(\xb)d\mu(\xb) - \int_\Yc \varphi_\Phi(G_\Theta(\yb))d\nu(\yb)\right) \\
&+ (\left(\int_\Yc \psi_\Xi(\yb)d\nu(\yb) - \int_\Xc \psi_\Xi(\Hc \xb)d\mu(\xb)\right) \ .
\end{split}
\label{eq:WGAN loss}
\end{eqnarray} 
Here, Kantorovich 1-Lipschitz potential $\varphi_\Phi$ and $\psi_\Xi$ correspond to the Wasserstein GAN discriminators, where a function $f:\xb\in \Xc\mapsto f(\xb) \in \Rd$ is called 1-Lipschitz if it satisfies
\begin{align}
|f(\xb)-f(\yb)|\leq \|\xb-\yb\|\quad\forall \xb,\yb\in \Xc.
\end{align}
If $\Hc$ is known, we further showed that the maximization of the discriminator $\psi_\Xi$ in \eqref{eq:WGAN loss} with respect to $\Xi$ does not affect the generator $G_\Theta$.
Therefore, the discriminator $\psi_\Xi$ can be neglected\cite{9136890}.
This consideration leads to a novel OT-cycleGAN architecture with only one generator and one discriminator. 
Specifically, our OT-cycleGAN cost function can be reformulated as
\begin{eqnarray}
\min_\Theta \Kd(\Theta) = \min_\Theta \max_\Phi \ell(\Theta;\Phi) \ ,
\end{eqnarray}
where $\ell(\Theta;\Phi) = \gamma\ell_{cycle}(\Theta) + \ell_{WGAN}(\Theta;\Phi)$. 
Here, $\ell_{cycle}$ and $\ell_{WGAN}$ are simplified as
\begin{eqnarray}
\begin{split}
\ell_{cycle}(\Theta) =  \int_\Xc \|\xb - G_\Theta (\Hc \xb)\|d\mu(\xb) \\
+  \int_\Yc \|\yb-\Hc G_\Theta (\yb)\| d\nu(\yb) \ ,
\end{split}
\label{eq:our cycle loss}
\end{eqnarray}
\begin{eqnarray}
\begin{split}
&\ell_{WGAN}(\Theta; \Phi) \\
&= \left(\int_\Xc \varphi_\Phi(\xb)d\mu(\xb) - \int_\Yc \varphi_\Phi(G_\Theta (\yb))d\nu(\yb)\right) \ .
\end{split}
\label{eq:our WGAN loss}
\end{eqnarray}

In contrast to the recent inverse problem approaches using deep learning prior\cite{zhang2017learningdenoiser,aggarwal2018modl}, which incorporate physics-driven information during the run-time reconstruction, our cycleGAN formulation shows that the physics-driven constraint can be incorporated in the training phase of a feed-forward deep neural networks so that the trained network can generate physically meaningful estimates in a real time manner.

\subsection{CycleGAN for Accelerated MRI}
We are now ready to use the theory to derive unique OT-cycleGAN architectures for accelerated MR acquisition.
As the single coil case is considered a special case of the parallel imaging, we derive our formula for the parallel imaging problem for simplicity.

In accelerated MRI for the multi-channel acquisition, the forward measurement model can be described as
\begin{eqnarray}\label{eq:forward}
\hat \Xb = \Pc_\Omega \Fc \Xb \ ,
\end{eqnarray}
where
\begin{align*}
\Xb&=[\xb^{(1)},\cdots, \xb^{(C)}] \ ,\\
\widehat \Xb&=[\widehat \xb^{(1)},\cdots, \widehat\xb^{(C)}] \ .
\end{align*}
Here, the superscript $^{(i)}$ denotes the $i$-th coil and $C$ is the number of coils, and $\xb^{(i)}$ is the fully sampled image, $\hat \xb^{(i)}$ is the downsampled $k$-space, $\Fc$ is the 2-D Fourier transform, and $\Pc_\Omega$ is the projection to $\Omega$ that denotes $k$-space sampling indices. 

In order to implement every step of the neural network in the image domain without alternating between the $k$-space and the image domain, this forward measurement is converted to the image domain by applying an inverse 2-D Fourier transform, $\Fc^{-1}$:
\begin{eqnarray}\label{eq:model}
\Yb = \Fc^{-1}\Pc_\Omega \Fc \Xb \ ,%+ \Fc^{-1}\wb
\end{eqnarray}
where
\begin{align*}
\Yb&=[\yb^{(1)},\cdots, \yb^{(C)}] \ .
\end{align*}
A blind extension of \eqref{eq:cost0} to the multi-channel MRI reconstruction problem would lead to the following PLS cost:
\begin{eqnarray}\label{eq:costp}
c(\Xb,\Yb;\Theta)=\sum_{i=1}^C\|\yb^{(i)}-\Hc\xb^{(i)}\|+ \| G_{\Theta}(\yb^{(i)}) -\xb^{(i)}\| \ ,
\end{eqnarray}
where 
\begin{align*}
\Hc:= \Fc^{-1} \Pc_\Omega \Fc \ .
\end{align*}
The main problem of this approach is that the dual formulation leads to a complicated discriminator that should compare all coil images channel by channel with the corresponding full resolution channel images.
Since each coil image is highlighted with a specific sensitivity pattern, we found that the resulting cycleGAN is mainly trained to recover sensitivity patterns rather than underlying images.
 
To address this issue, we propose the following cost function for PLS formulation as a transport cost:
\begin{align}
c(\Xb,\Yb;\Theta) = & \| \Mc \Yb - \Mc \Hc\Xb\| \notag\\
&+ \|\Mc G_\Theta(\Yb) - \Mc \Xb\| \ ,
\label{eq:costMR}
\end{align}
where $\Mc: \Xc \times \cdots\times \Xc \mapsto \Xc$ is the element-wise square-root of sum-of-the squares (SSoS) operation such that $\zb=\Mc(\Xb)$ is composed of elements
\begin{align}\label{eq:sos}
z_n =  \left(\sum_{i=1}^C |x^{(i)}_n|^2 \right)^{\frac{1}{2}} \ ,
\end{align}
where $x_n^{(i)}$ is the $n$-th element of the $i$-th coil image.
One of the main advantages of our cost \eqref{eq:costMR} over \eqref{eq:costp} is that the image error is calculated using SSoS images after removing the coil sensitivity dependency on each
channel image.
The resulting discriminator also compares the true and fake images using SSoS images so that the generated image quality is determined by the underlying images rather than the coil sensitivity pattern.

Since the forward operator $\Hc$ is known, our cycleGAN architecture has only one generator and one discriminator as discussed before.
Specifically, our cycleGAN cost function can be reformulated as
\begin{eqnarray}
\min_\Theta \Kd(\Theta) = \min_\Theta \max_\Phi \ell(\Theta;\Phi) \ ,
\end{eqnarray}
where $\ell(\Theta;\Phi) = \gamma\ell_{cycle}(\Theta) + \ell_{WGAN}(\Theta;\Phi)$.
Here, $\ell_{cycle}$ loss is calculated using SSoS images as:
\begin{align}
\ell_{cycle}(\Theta) &=  \int_{\Xc\times \cdots \times \Xc}\|\Mc\Xb - \Mc G_\Theta (\Fc^{-1}\Pc_\Omega \Fc \Xb)\|d\mu(\Xb) \notag \\
&+  \int_{\Yc\times \cdots \times \Yc} \|\Mc\Yb-\Mc\Fc^{-1}\Pc_\Omega \Fc G_\Theta (\Yb)\| d\nu(\Yb) \ ,
\label{eq:our cycle loss parallel}
\end{align}
and the GAN loss $\ell_{WGAN}$ is also simplified as
\begin{align}
&\ell_{WGAN}(\Theta; \Phi) =  \notag\\
&\int_{\Xc\times \cdots \times \Xc} \varphi_\Phi(\Xb)d\mu(\Xb) - \int_{\Yc\times \cdots \times \Yc}  \varphi_\Phi(G_\Theta (\Yb))d\nu(\Yb) \ ,
\label{eq:our WGAN loss parallel}
\end{align}
where the discriminator $\varphi_\Phi(\Xb)$ should be 1-Lipschitz with respect to multi-coil images to satisfy the mathematical validity of the derivation in our companion paper\cite{9136890}.
This implies that we should show that
\begin{align}\label{eq:1lips}
|\varphi_\Phi(\Xb)-\varphi_\Phi(\Xb') | & \leq \|\Xb-\Xb'\|_F,\\
&\quad \mbox{for all} \quad \Xb,\Xb'\in \Xc\times \cdots\times \Xc \notag \ ,
\end{align}
where $\|\cdot\|_F$ denotes the Frobenius norm.
However, as discussed before, care should be taken for the designing the discriminator to satisfy the 1-Lipschitz property \eqref{eq:1lips}.
For example, if the discriminator compares each coil images separately, the network architecture is quite complicated.
In addition, each coil image is dominated by the coil sensitivity patterns, making the network difficult to train.

One of the important contributions of this paper is, therefore, to show that there exists a simpler discriminator architecture that is robust to distinct coil sensitivity patterns by exploiting the SSoS image.
Specifically, Proposition~\ref{prp:main2} shows that 1-Lipschitz property of the original Kantorovich potential $\varphi_\Phi$ can be retained by a discriminator using SSoS image if Kantorovich potential is decomposed as
\begin{align}\label{eq:class_parallel}
\varphi_\Phi(\Xb) = \left(\theta_\Phi\circ \Mc\right)(\Xb) \ ,% \mathrm{Tr}(xx^\top)\right)
\end{align}
where $\circ$ denotes the composite function, $\Mc$ is the SSoS operation, and $\theta_\Phi$ is a 1-Lipschitz function with respect to the SSoS image.
Therefore, the resulting cycleGAN architecture is implemented as the SSoS operation as shown in Fig.~\ref{fig:cycleGAN}, so that it just requires to tell the difference between the fake and true images in the SSoS image domain.
\begin{proposition}\label{prp:main2}
If $\theta_\Phi$ in \eqref{eq:class_parallel} is 1-Lipschitz in the SSoS domain, then $\varphi_\Phi$ in \eqref{eq:class_parallel} is also 1-Lipschitz, i.e. it satisfies \eqref{eq:1lips}.
\end{proposition}
\begin{proof}
Let $N$ and $C$ denote the number of image pixel and coils, respectively.
For $\Xb\in \Rd^{N\times C}$,  the SSoS image is given by
\begin{align*}
\Mc(\Xb)&=\begin{bmatrix} \left(\sum_{i=1}^C |x^{(i)}_1|^2 \right)^{\frac{1}{2}} & \cdots & \left(\sum_{i=1}^C |x^{(i)}_N|^2 \right)^{\frac{1}{2}}  \end{bmatrix}^\top\\
&=\begin{bmatrix} \|\Xb^1\| & \cdots & \|\Xb^N\| \end{bmatrix}^\top \ ,
\end{align*}
where $\Xb^i$ denote the $i$-th row of $\Xb$.
Accordingly, we have
\begin{align}
\|\Mc(\Xb)-\Mc(\Zb)\| %&=   \sum_n\left(\left(\sum_{i=1}^C |x^{(i)}_n|^2 \right)^{\frac{1}{2}}- \left(\sum_{i=1}^C |x^{(i)}_n|^2 \right)^{\frac{1}{2}} \right)\\
&=\left(\sum_{i=1}^N \left( \|\Xb^i\| - \|\Zb^{i}\| \right)^2\right)^{\frac{1}{2}} \ ,
\end{align}
Thanks to the triangular inequality for the row vectors:
$$\left( \|\Xb^i\| - \|\Zb^{i}\| \right)^2 \leq \|\Xb^i - \Zb^{i}\|^2 ~  \ ,$$
we have
\begin{align}
\|\Mc(\Xb)-\Mc(\Zb)\| &\leq  \left(\sum_{i=1}^N\|\Xb^i - \Zb^{i}\|^2 \right)^{\frac{1}{2}}  \notag\\
&= \|\Xb-\Zb\|_F \ . \label{eq:mid}
\end{align}
Using \eqref{eq:mid} and the assumption that $\theta_\Phi$ is 1-Lipschitz, we have
\begin{align*}
|\varphi_\Phi(\Xb)-\varphi_\Phi(\Zb)| &= |  \left(\theta_\Phi\circ \Mc\right)(\Xb) -  \left(\theta_\Phi\circ \Mc\right)(\Zb)| \\
&\leq \|\Mc(\Xb)-\Mc(\Zb)\| \notag\\
& \leq  \|\Xb-\Zb\|_F \ .
\end{align*}
where the first inequality comes from 1-Lipschitz property of $\theta_\Phi$ for the SSoS images, and the last inequality comes from \eqref{eq:mid}.
Therefore, $\varphi_\Phi$ satisfies \eqref{eq:1lips} and is 1-Lipschitz.
\end{proof}

To ensure that the Kantorovich potential becomes 1-Lipschitz in the SSoS domain, we use Wasserstein GAN loss in (\ref{eq:our WGAN loss}) with the gradient penalty loss\cite{gulrajani2017improved}.
Then, our final WGAN loss functions is given by:
\begin{align}\label{eq:eta}
&\ell_{WGAN}(\Theta; \Phi) =  \notag\\
&\int_{\Xc\times \cdots \times \Xc} \varphi_\Phi(\Xb)d\mu(\Xb) - \int_{\Yc\times \cdots \times \Yc}  \varphi_\Phi(G_\Theta (\Yb))d\nu(\Yb) \notag \\
&+\eta \int_{\tilde{\Xc}}(\|\nabla_{{\zb}}\theta_\Phi(\zb)\|_2 - 1)^2d\tau
(\zb) \ ,
\end{align}
where $\varphi_\Phi(\Xb) = \left(\theta_\Phi\circ \Mc\right)(\Xb)$ and $\eta>0$ is a regularization parameter, and $\tau(\zb)$ is a measure for the SSoS images $\Mc(\xb),\xb\in\Xc$.
Here, $\zb=\alpha \Mc(\Xb)+(1 - \alpha)\Mc\left(G_\Theta(\Yb)\right)$, where $\alpha$ is a random variable from the uniform distribution between $[0,1]$\cite{gulrajani2017improved}.

\section{Method}\label{sec:method}
\subsection{Dataset}
To verify the generalizability of our algorithm, we use two datasets, the human connectome project (HCP) data and the fastMRI \cite{zbontar2018fastmri} data, which have different characteristics.
Note that the $k$-space data from the first dataset are obtained by simulation, whereas the second dataset is acquired from real scanner using a retrospective subsampling with the sampling masks that are used in practice.
The details are as follows.

\subsubsection{Human Connectome Project (HCP) Dataset} 
Human connectome project (HCP) data are composed of human brain MR images. 
This dataset was acquired by Siemens 3T system, with 3D spin echo protocol. 
The parameters for HCP data acquisition are as follows: echo train duration = 1105, matrix size = $320 \times 320$, voxel size = 0.7mm $\times$ 0.7 mm $\times$ 0.7 mm, TR = 3200 ms, and TE = 565 ms.
This dataset only contains magnitude images. 
To acquire multi-coil MR images, we estimate coil sensitivity maps by simulation (http://bigwww.epfl.ch/algorithms/mri-reconstruction/).
The number of coils ($C$) was 8 for HCP data.

To train our network, we use 3200 brain images, while the other 600 slices were used to test our model. 
We adopt 2D Cartesian random sampling pattern for HCP dataset, so $k$-space samples are acquired randomly from the uniform distribution. 
Also, we set the acceleration factor $R$ as 4, and $36 \times 36$ size of the central $k$-space samples are included as the autocalibration signal (ACS).
An example of the sampling mask for HCP data is shown in Fig. \ref{fig:sampling mask}(a).

\subsubsection{FastMRI Dataset}
FastMRI data \cite{zbontar2018fastmri} are provided by the NYU Langone Health for MR reconstruction challenge. 
Multi-coil fully sampled MRI data were acquired from either three different clinical 3T systems or one clinical 1.5T system. 
The number of coils ($C$) was 15, and 2D Cartesian turbo spin echo (TSE) protocol was used. 
The dataset contains coronal proton density weighting with fat suppression (PDFS) and without fat suppression (PD) data.
The parameters for data acquisition are as follows: echo train length (ETL) = 4, matrix size = $320 \times 320$, in-plane resolution = 0.5 mm $\times$ 0.5 mm, slice thickness = 3 mm, TR = 
2200 - 3000 ms, TE = 27 - 34 ms. 
Also, to simulate single coil data from multi-coil data, \cite{zbontar2018fastmri} used emulated single coil (ESC) methodology. 

The fastMRI database in \cite{zbontar2018fastmri} also provide training, validation, and masked test sets, and these datasets include $k$-space data and MR images of human knees. 
The height of images is 640, but the width differs. 
For our experiments, we extract 3500 MR images from both single and multi-coil validation data. 
Then, 3000 slices out of the total 3500 are used for training, and the rest of 500 slices are used for testing. 
Training data and test data are extracted from different subjects.
To acquire downsampled MR data, we use 1D Cartesian uniform random sampling pattern in fastMRI data as shown in Fig. \ref{fig:sampling mask}(b), (c). 
We set the acceleration facter $R$ as 4 or 6. 
The ACS region contains 8\% of central $k$-space lines when $R$ is 4, and 6\% of central $k$-space lines when $R$ is 6. 

\begin{figure}[!ht]
	\centering
	\begin{minipage}[b]{0.3\linewidth}
		\centerline{\includegraphics[width=\linewidth,height=2.8cm]{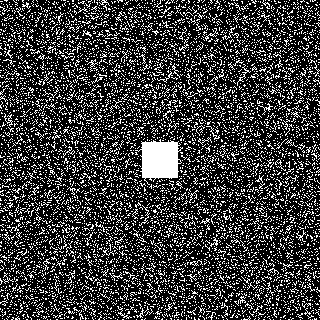}}
		\centerline{(a)}\medskip
	\end{minipage}
	\hspace{0.2cm}
	\begin{minipage}[b]{0.3\linewidth}
		\centerline{\includegraphics[width=\linewidth,height=2.8cm]{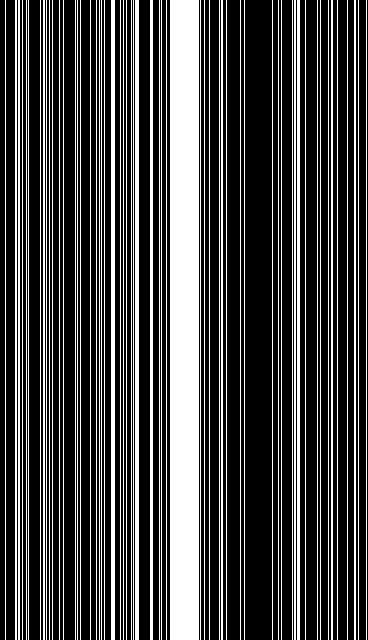}}
		\centerline{(b)}\medskip
	\end{minipage}
	\hspace{0.2cm}
	\begin{minipage}[b]{0.3\linewidth}
		\centerline{\includegraphics[width=\linewidth,height=2.8cm]{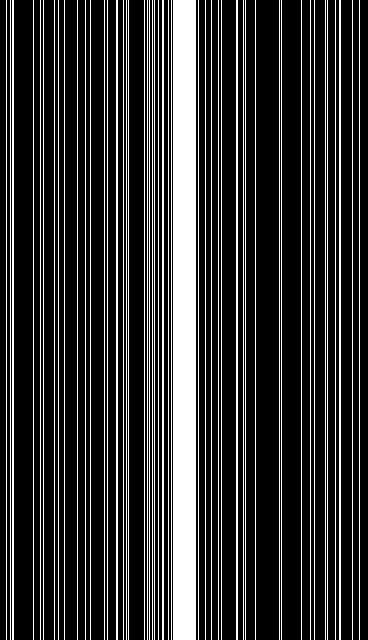}}
		\centerline{(c)}\medskip
	\end{minipage}
	\caption{Examples of sampling masks: (a) 2D Cartesian random sampling with $R = 4$, (b) 1D Cartesian random sampling with $R = 4$,
	and  (c) 1D Cartesian random sampling with $R = 6$.}
	\vspace{-0.5cm}
	\label{fig:sampling mask}
\end{figure}

\subsection{Network Architecture}
The schematic diagram of the proposed cycleGAN architecture is illustrated in Fig. \ref{fig:cycleGAN}.
Although the conventional cycleGAN has two generators and two discriminators for translation between two domains, our OT-cycleGAN architecture has only one generator and one discriminator due to the deterministic forward path.
Accordingly, the network training is more stable, and the training time is reduced by a factor of two compared to the conventional cycleGAN.

Note that there exist no paired images for the update of the generator and the discriminator, so the images used for the top and bottom in Fig.~\ref{fig:cycleGAN} are different.  
More specifically, the high-resolution domain is for fully sampled images, so all downsampled images for the top branch of our model were generated from the fully sampled images with random sampling masks.  
Furthermore, sampling masks in the bottom branch are randomly generated for every step with fixed acceleration factor and ratio of ACS region to improve the generalization capability of the algorithm.
Additionally, different undersampling masks were used to different MR volumes in fastMRI challenges, so it is reasonable to generate random sampling masks for every step so that one network can be used for various dataset. 
For this reason, the two sampling masks in Fig.~\ref{fig:cycleGAN} are different.

We use U-Net \cite{ronneberger2015u} generator to reconstruct fully sampled images from downsampled images as shown in Fig. \ref{fig:G and D}(a). 
Our generator has an encoder part and a decoder part. 
The encoder part consists of $3 \times 3$ convolution, instance normalization\cite{ulyanov2016instance}, and leaky ReLU. 
In the decoder part, the nearest neighbor upsampling and $3 \times 3$ convolution are used. 
Also, there are max pooling layers and skip-connection through the channel concatenation. 
To handle complex values, we concatenate real values and imaginary values along the channel dimension. 
Consequently, an input and an output of the generator have 2$C$ channels, where $C$ denotes the number of coils. 
Once the complex images are generated, the cycle-consistency loss are calculated using the square-root of sum-of-the squares (SSoS) images.

Our discriminator is same as the discriminator of the original cycleGAN. 
We use PatchGAN\cite{isola2017image} discriminator as shown in Fig. \ref{fig:G and D}(b). 
PatchGAN discriminator classifies inputs at patch scales, and it thus helps the generator to reconstruct the high frequency structure of images. 
The discriminator has convolution layers, instance normalization, and leaky ReLU operations. 
As discussed before, inputs of the discriminator consist of SSoS images.
For single coil images, the SSoS operation leads to the magnitude images.

\begin{figure}[!hbt]
	\centering
	\begin{minipage}[b]{\linewidth}
		\centerline{\includegraphics[width=\linewidth]{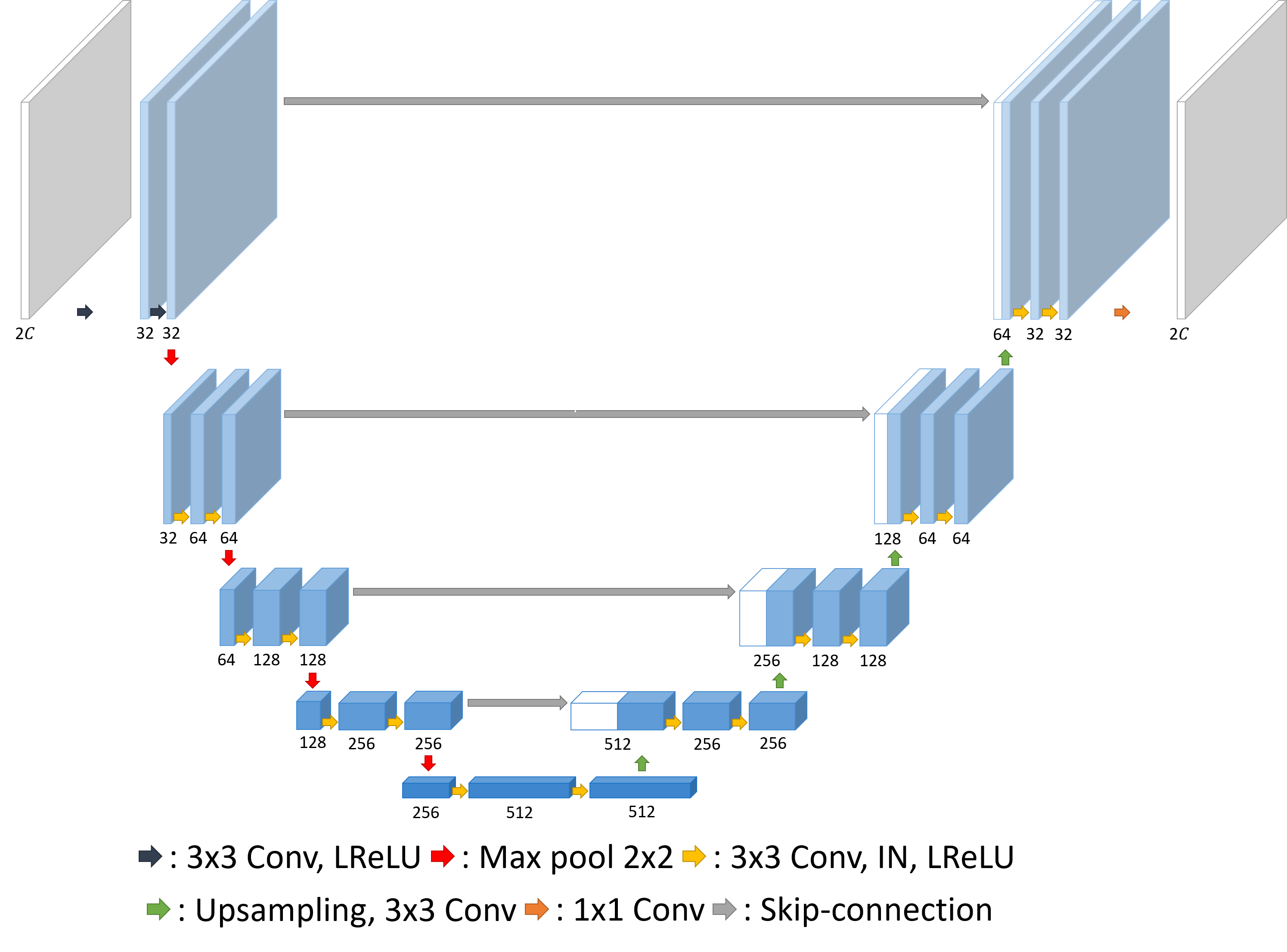}}
		\centerline{(a)}\medskip
	\end{minipage}
	\begin{minipage}[b]{0.9\linewidth}
		\centerline{\includegraphics[width=\linewidth]{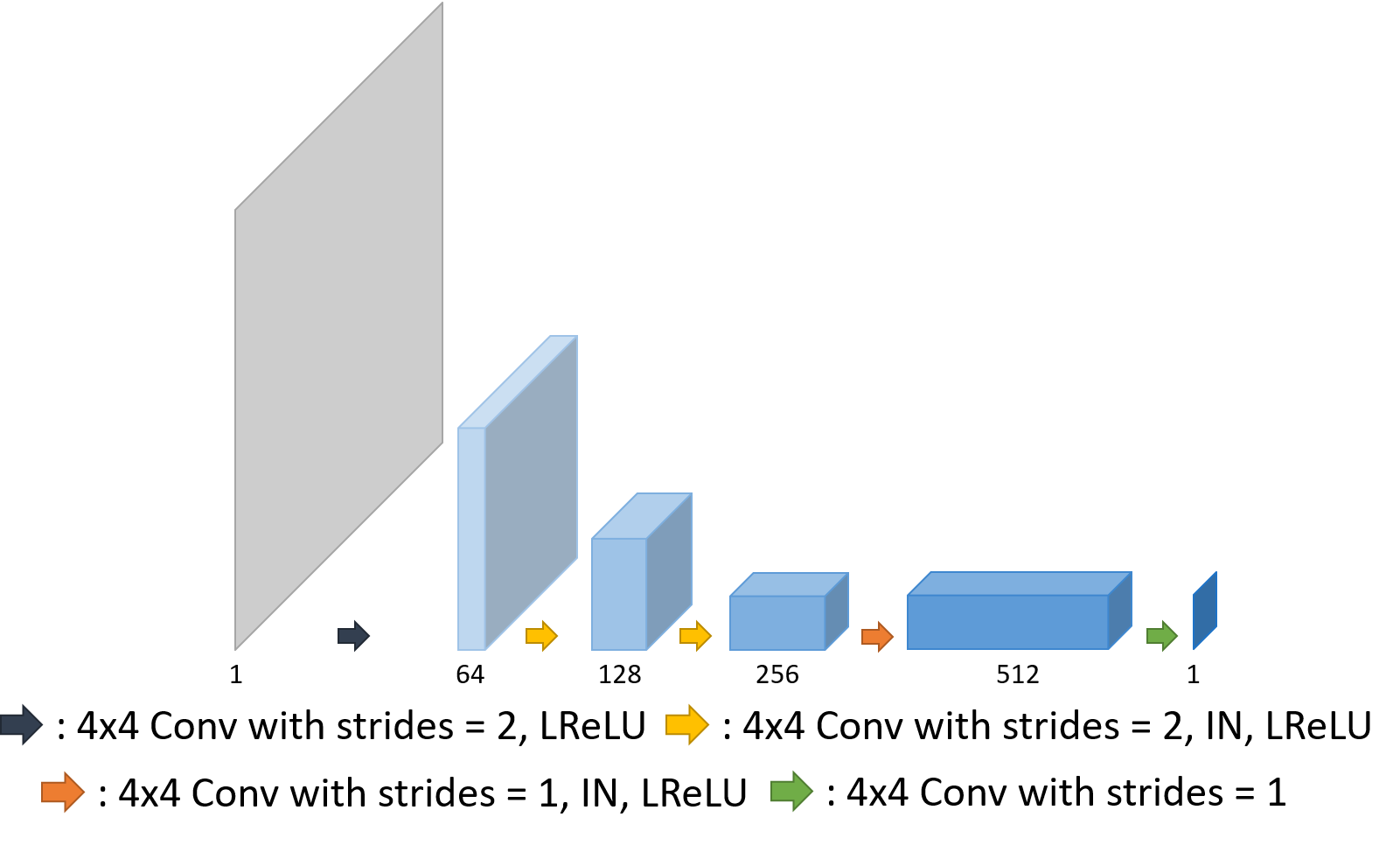}}
		\centerline{(b)}\medskip
	\end{minipage}
	\caption{Network architecture for (a) our generator and (b) our discriminator.}
	\vspace*{-0.5cm}
	\label{fig:G and D}
\end{figure}

\subsection{Quantitative Metrics}
For the quantitative evaluation, we use peak signal-to-noise ratio (PSNR) and structural similarity index metric (SSIM).
PSNR is defined through the mean squared error (MSE):
\begin{eqnarray}
MSE = \frac{1}{N}\| \zb-\hat \zb\|^2 \ ,
\label{eq:psnr}
\end{eqnarray}
where $N$ is the number of pixels, and  $\zb$ is the ground truth SSoS image, and $\hat\zb$ is the reconstruction SSoS image.
Then, PSNR can be defined as follows:
\begin{eqnarray}
PSNR = 20\log_{10}\left(\frac{MAX_\zb}{\sqrt{MSE}}\right) \ ,
\end{eqnarray}
where $MAX_\zb$ is the maximum possible pixel value of the image.
Next, SSIM is defined as follows:
\begin{eqnarray}
SSIM = \frac{(2\mu_\zb\mu_{\hat\zb} + c_1)(2\sigma_{\zb,\hat\zb} + c_2)}
{(\mu_\zb^2 + \mu_{\hat\zb}^2 + c_1)(\sigma_\zb^2 + \sigma_{\hat\zb}^2 + c_2)} \ ,
\end{eqnarray}
where $\mu_\zb$, $\mu_{\hat\zb}$ are the average of each image, $\sigma_\zb$, $\sigma_{\hat\zb}$ are the variance of each image, $\sigma_{\zb,\hat\zb}$ is the covariance of images, $c_1 = (k_1L)^2$, $c_2 = (k_2L)^2$ are variables to stabilize the division with weak denominator, where $L$ is the dynamic range of pixel values, and $k_1 = 0.01$, $k_2 = 0.03$ by default.

\begin{figure*}[!t] 	
	\centerline{\includegraphics[width=0.95\linewidth]{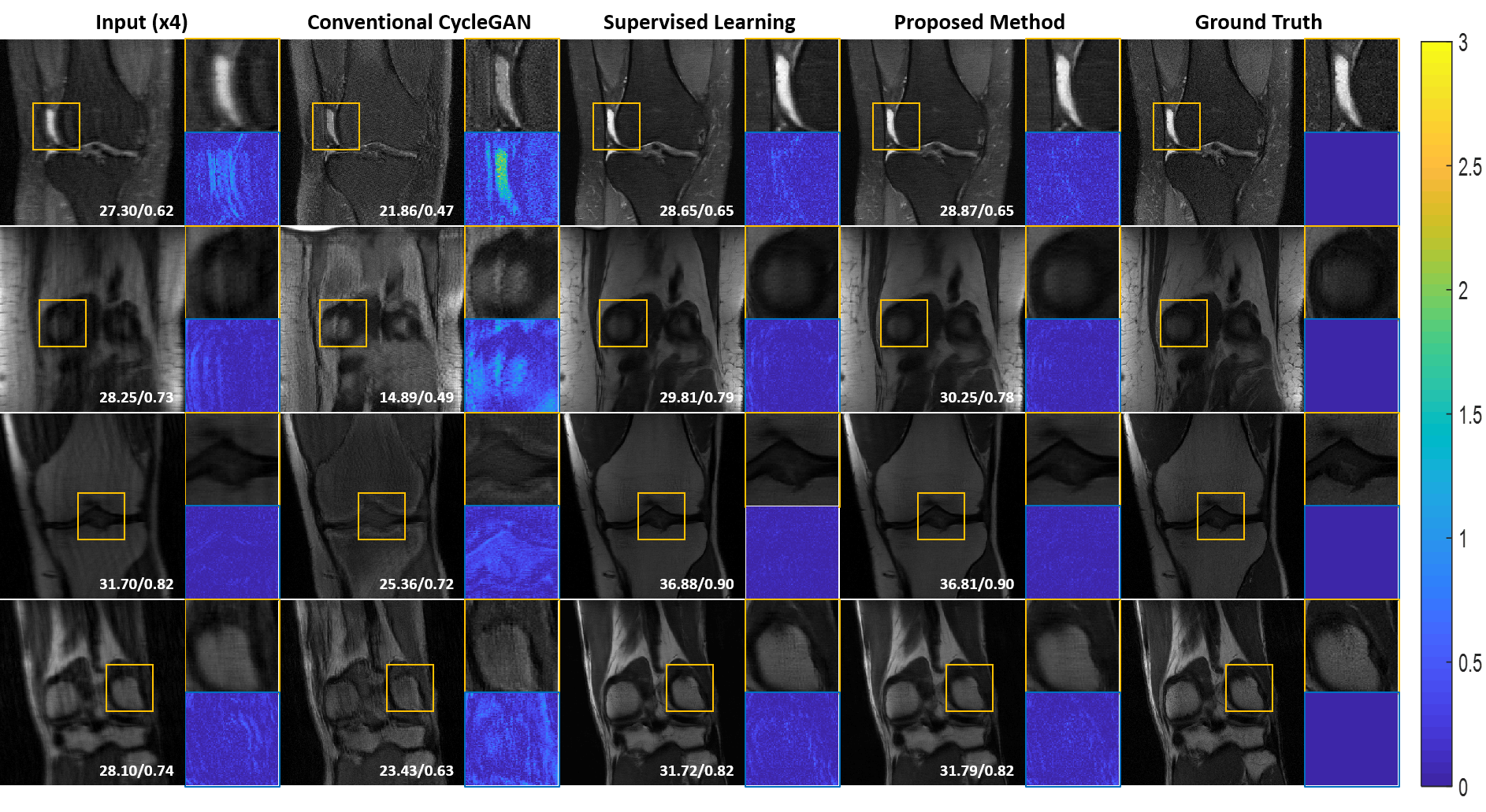}}
	\caption{The reconstruction results of fastMRI single coil ($R = 4$) data using 
		paired/unpaired learning. 
		The values in the corner are PSNR/SSIM values of each slice.
		The difference images are amplified by a factor of five.
		The color bar of the difference images is at the right of the figure.}
	\label{fig:fast_single_4_result}
\end{figure*}

\subsection{Training Details}
Each MR image is normalized by the standard deviation of magnitude values of downsampled images. 
For HCP data experiments, objective functions and quantitative metrics are calculated for the full size of MR images.
On the other hand, only the center $320 \times 320$ is used in the fastMRI evaluation protocol because the fastMRI data were oversampled along the readout direction with signal suppression pulse to prevent overlapping from wrap-around artifact, so most of the values in periphery are zero or do not contain meaningful structures.
Therefore, reconstructed images are cropped at the center by the size of $320 \times 320$, and objective functions are computed for the center cropped images.
After the reconstruction, quantitative metrics are also measured at the center of images.

In our implementation, we minimize (\ref{eq:eta}) with $\eta = 10$.
To optimize our model, we use Adam optimizer with momentum parameters $\beta_1 = 0.5$, $\beta_2 = 0.9$, batch size of 1, and learning rate of 0.0001. 
The discriminator is updated 5 times per each update of the generator. 
Also, our network is trained for 100 epochs. 
Our code is implemented with TensorFlow.

\subsection{Baseline Methods}
We use supervised learning and the standard cycleGAN\cite{zhu2017unpaired} as baseline methods to validate the performance of our method. 
For supervised learning, we use the same network in Fig. \ref{fig:G and D}(a), and $\ell_1$-loss as an objective function. 
We train supervised method for 100 epochs with the learning rate of 0.0001. 
In the standard cycleGAN\cite{zhu2017unpaired}, there are two generators and two discriminators, and the structure of the neural networks are the same as our generator and discriminator. 

\section{Experimental Results}\label{sec:result}
\subsection{Single Coil Experiments}
Fig. \ref{fig:fast_single_4_result} shows the reconstruction results of fastMRI single coil data using various methods when the acceleration factor is 4. 
As shown in Fig. \ref{fig:fast_single_4_result}, there are aliasing artifacts in the input images due to the subsampling of $k$-space. 
When the conventional cycleGAN was used, the reconstructed images are distorted and the aliasing pattern still remains in the reconstructed images. 
Also, the corresponding values of quantitative metrics were even inferior to input images because of the degradation of quality in the reconstructed images.
This implies that the conventional cycleGAN fails in this application.
On the other hand, the supervised learning method successfully reconstructed high quality MR images. 
Interestingly, our proposed method produces compatible reconstruction results to the supervised learning method,
even providing better images in some cases, preserving the 
texture of original images.

Table \ref{tbl:fast metric} shows average quantitative metric values of various reconstruction algorithms. 
We confirm that the proposed method shows competitive quantitative results compared to the supervised method in single coil experiment.

\begin{figure*}[!t] 	
	\centerline{\includegraphics[width=0.95\linewidth]{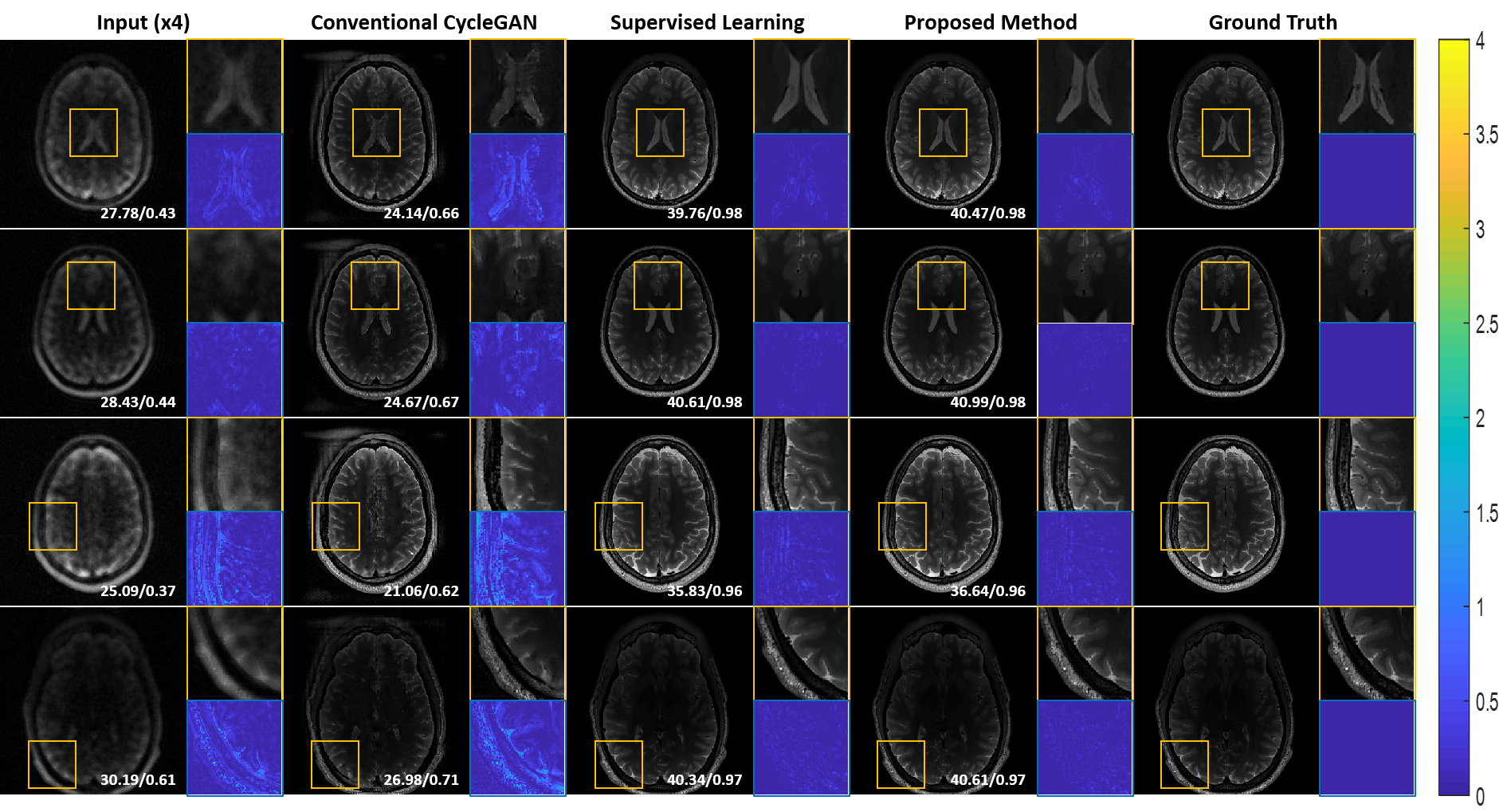}}
	\caption{The reconstruction results of HCP multi-coil ($R = 4$) data using 
		paired/unpaired learning. 
		The values in the corner are PSNR/SSIM values of each slice.
		The difference images are amplified by a factor of five.
		The color bar of the difference images is at the right of the figure.}
	\label{fig:HCP_multi_4_result}
\end{figure*}

\begin{figure*}[!t] 	
	\centerline{\includegraphics[width=0.95\linewidth]{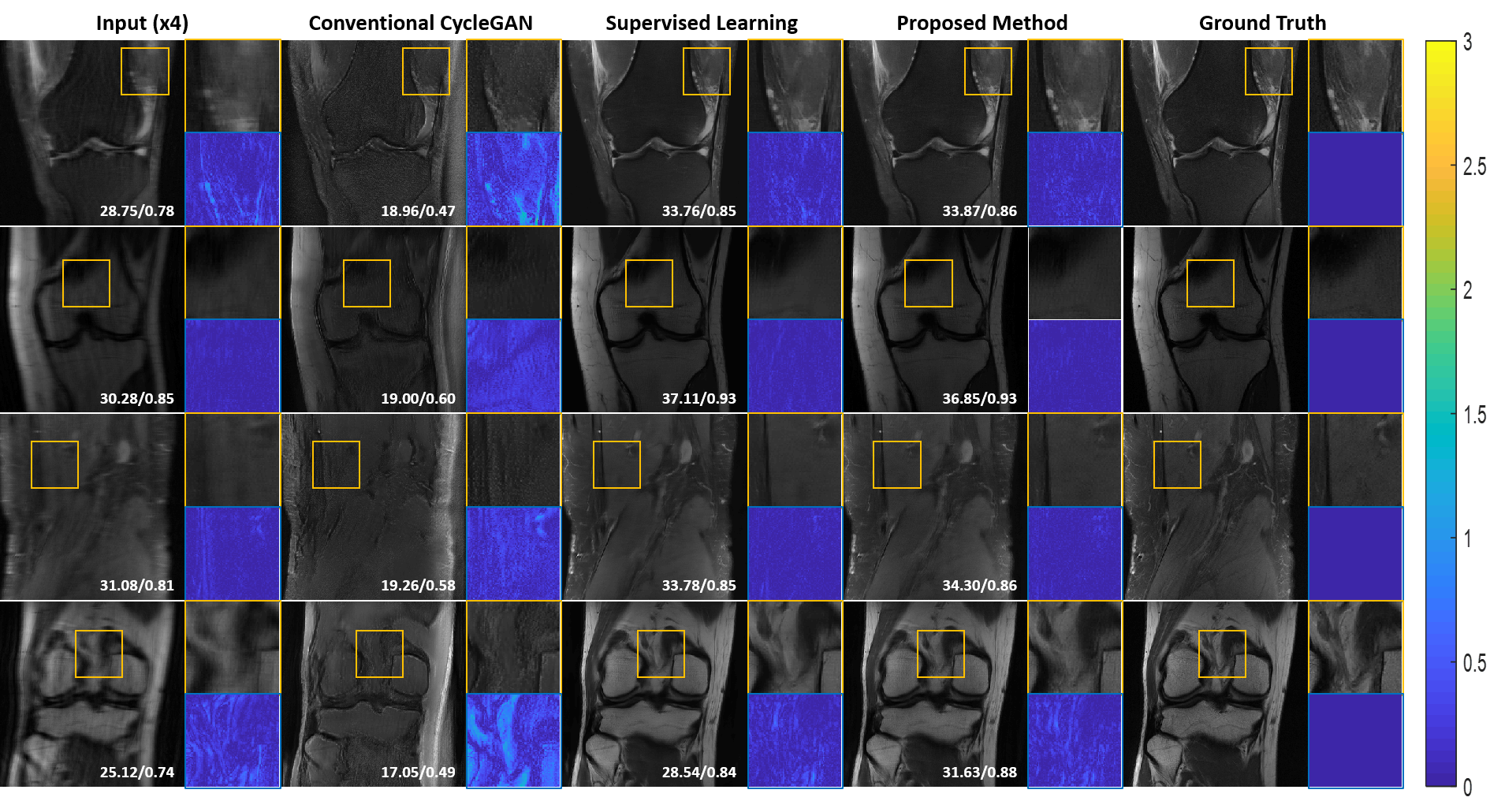}}
	\caption{The reconstruction results of fastMRI multi-coil ($R = 4$) data using 
		paired/unpaired learning. 
		The values in the corner are PSNR/SSIM values of each slice.
		The difference images are amplified by a factor of five.
		The color bar of the difference images is at the right of the figure.}
	\label{fig:fast_multi_4_result}
\end{figure*}

\begin{figure*}[!t] 	
	\centerline{\includegraphics[width=0.95\linewidth]{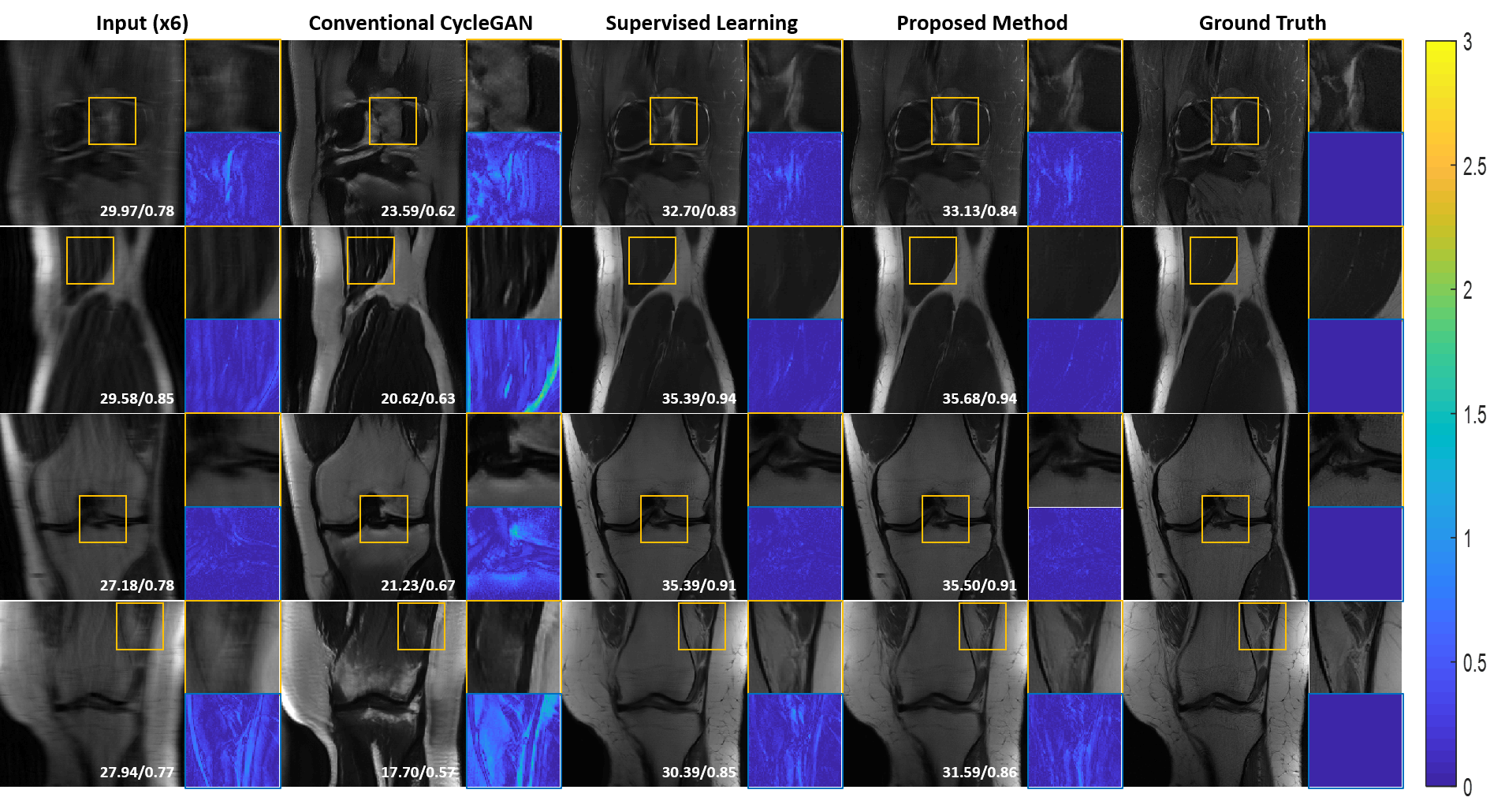}}
	\caption{The reconstruction results of fastMRI multi-coil ($R = 6$) data using 
		paired/unpaired learning. 
		The values in the corner are PSNR/SSIM values of each slice.
		The difference images are amplified by a factor of five.
		The color bar of the difference images is at the right of the figure.}
	\label{fig:fast_multi_6_result}
\end{figure*}

\begin{figure*}[!t] 	
	\centerline{\includegraphics[width=0.9\linewidth]{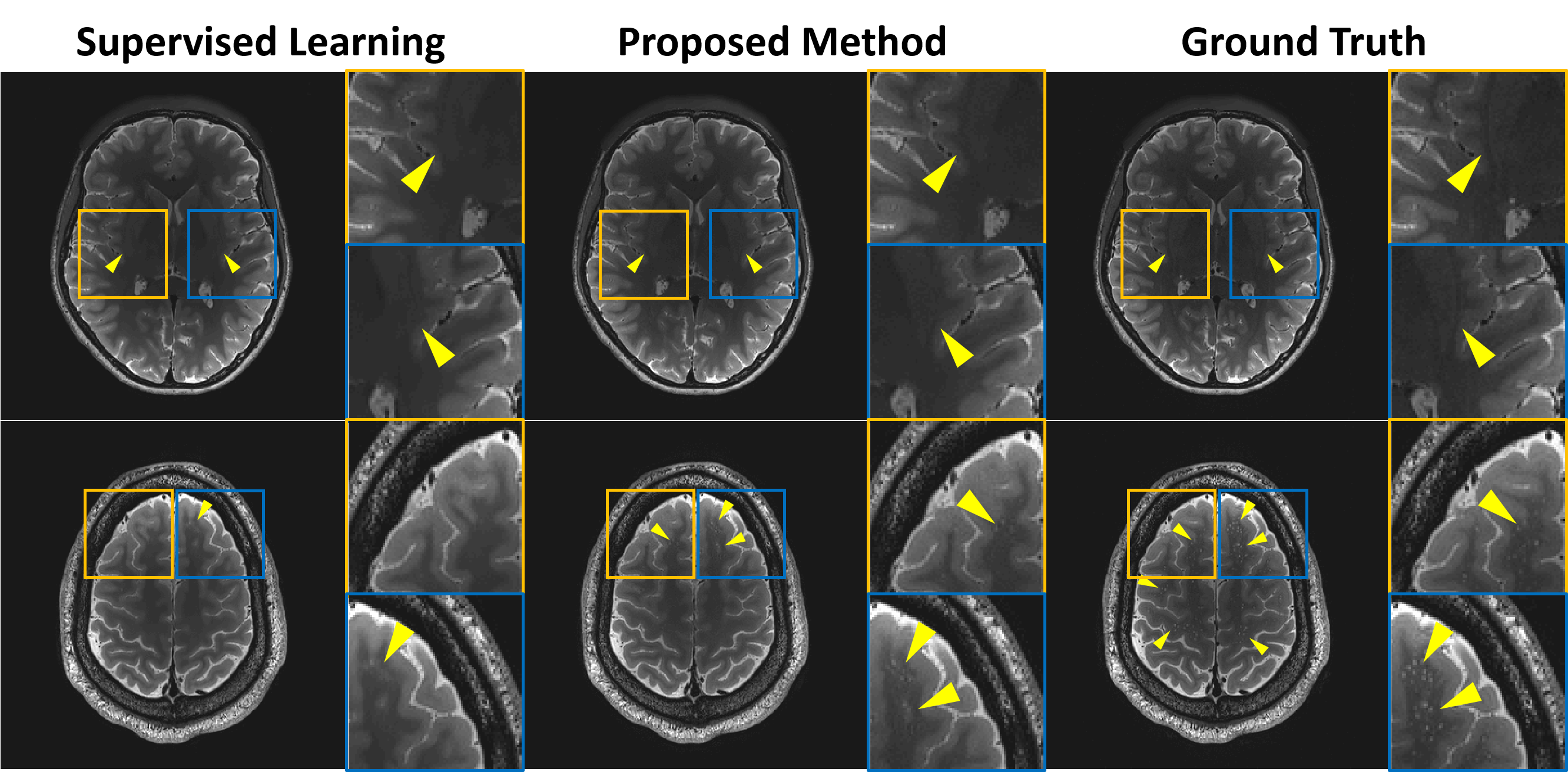}}
	\caption{Radiological evaluation for the reconstruction results of HCP multi-coil ($R = 4$) data using various methods.
	The window level of images is adjusted for better visualization.}
	\label{fig:radiological_evaluation}
\end{figure*}

\begin{table}[!hbt]
	\centering
	\caption{Quantitative comparison for various algorithms on fastMRI data. 
			The values in the table are average values for the whole test data.}
	\label{tbl:fast metric}
	\resizebox{0.4\textwidth}{!}{
		\begin{tabular}{c | c | c | c  c} 
			\toprule
			\multicolumn{3}{c|}{ } 														& PSNR (dB)			& SSIM					\\ \midrule\midrule
			\multirow{4}*{Single Coil}	& \multirow{4}*{$\times 4$}	& Input				& 26.8056			& 0.6279				\\
			\ 							&							& CycleGAN			& 20.8228			& 0.4881				\\
			\ 							&							& Supervised		& \textbf{28.7434}	& 0.6752				\\
			\ 							&							& Proposed			& 28.7426			& \textbf{0.6758}		\\ \midrule
			\multirow{8}*{Multi-Coil}	& \multirow{4}*{$\times 4$}	& Input				& 28.0524			& 0.7616				\\
			\							&							& CycleGAN			& 18.0951			& 0.4616				\\
			\							&							& Supervised		& \textbf{32.5213}	& \textbf{0.8315}		\\
			\							&							& Proposed			& 32.1812			& 0.8300				\\ \cmidrule{2-5}
			\							& \multirow{4}*{$\times 6$}	& Input				& 26.3487			& 0.7120				\\
			\							&							& CycleGAN 			& 20.2832			& 0.5537				\\
			\							&							& Supervised		& \textbf{31.1140}	& 0.8002				\\ 
			\							&							& Proposed			& 31.0229			& \textbf{0.8011}		\\ \bottomrule
		\end{tabular}
	}
\end{table}

\subsection{Multi-Coil Experiments}
Next, we verify the proposed method using multi-coil MR data. 
First, we conduct an experiment for HCP dataset. 
In Fig. \ref{fig:HCP_multi_4_result}, there are multi-coil reconstruction results for HCP data with $\times 4$ acceleration. 
Conventional cycleGAN generates artifacts around brains, and destroys details of the images. 
On the other hand, it is shown that the proposed method successfully reconstructs MR images with fine details. 
Moreover, as shown in Table~\ref{tbl:HCP multi 4 metric}, we verify that the quantitative metric values of the proposed method are competitive and even higher than that of the supervised method.

\begin{table}[!hbt]
	\centering
	\caption{Quantitative comparison for various algorithms on HCP data. 
			The values in the table are average values for the whole test data.}
	\label{tbl:HCP multi 4 metric}
	\resizebox{0.4\textwidth}{!}{
		\begin{tabular}{c | c | c | c  c}
			\toprule
			\multicolumn{3}{c|}{ } 														& PSNR (dB)					& SSIM					\\ \midrule\midrule
			\multirow{4}*{Multi-Coil}	& \multirow{4}*{$\times 4$}	& Input				& 26.6802					& 0.4332				\\
			\ 							&							& CycleGAN			& 23.0759					& 0.6384				\\
			\ 							&							& Supervised		& 37.3204					& 0.9570				\\
			\ 							&							& Proposed			& \textbf{37.7256}			& \textbf{0.9615}		\\ \bottomrule
		\end{tabular}
	}
\end{table}

Fig. \ref{fig:fast_multi_4_result} shows the fastMRI multi-coil reconstruction results from various methods with $R = 4$. 
Conventional cycleGAN still shows poor results in multi-coil data. 
Meanwhile, our proposed method provides accurate reconstruction with consistent texture faithfully emulating the characteristic of fully sampled MR images. 
Furthermore, we confirm that the proposed method produces comparable quantitative results to the supervised method as shown in Table \ref{tbl:fast metric}. 

We also apply our method for higher acceleration factor on multi-coil data. 
In Fig. \ref{fig:fast_multi_6_result}, the reconstruction results for $\times 6$ multi-coil images are shown. 
Our proposed method again successfully removes artifacts and recover fine details of MR images. 
Furthermore, the quantitative result of our method is close to the result of the supervised method  (Table \ref{tbl:fast metric}). 

\begin{figure*}[!t] 	
	\centerline{\includegraphics[width=0.99\linewidth]{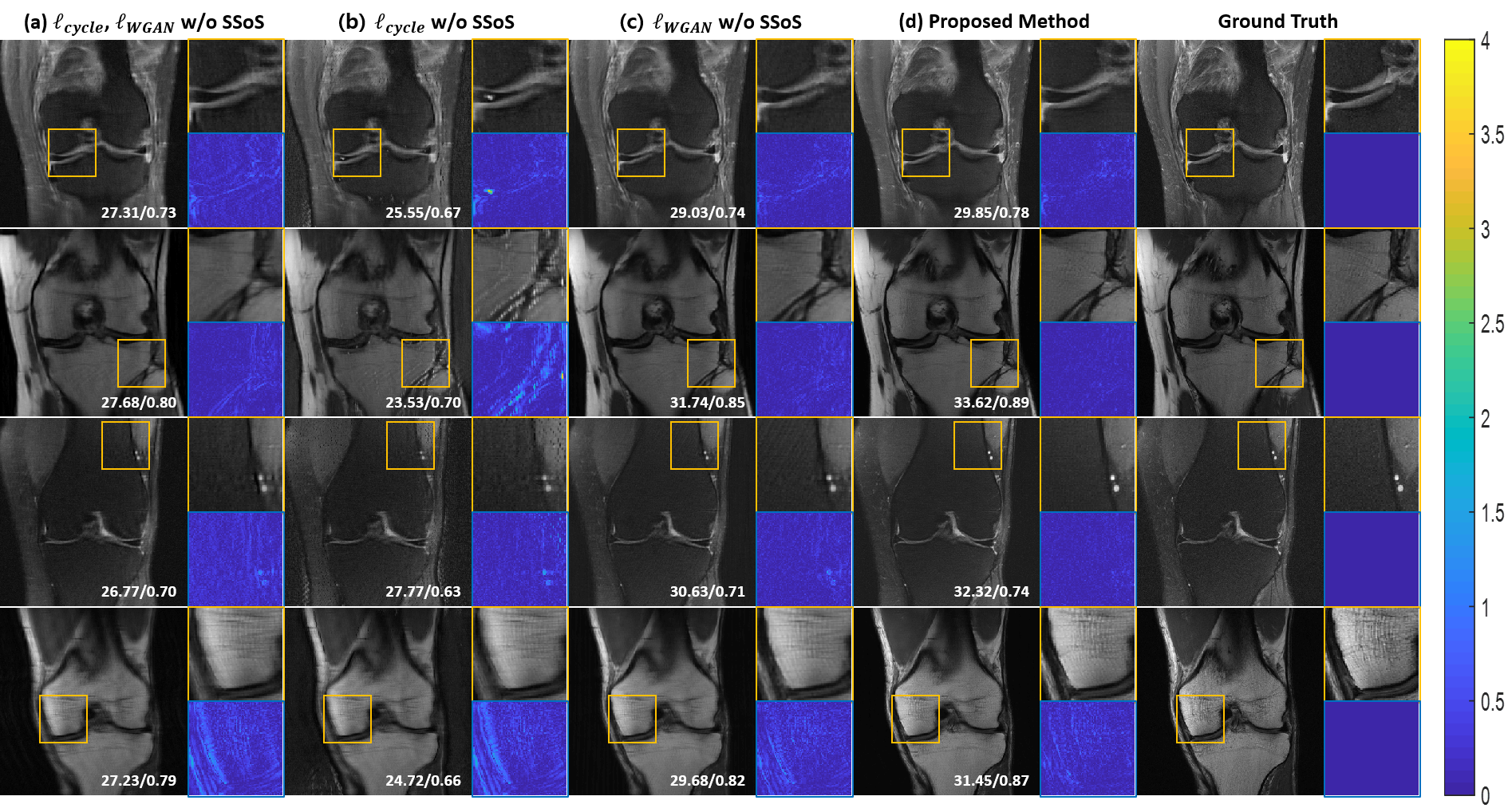}}
	\caption{The reconstruction results of fastMRI multi-coil ($R = 4$) data with or 
			without SSoS operation. 
			The values in the corner are PSNR/SSIM values of each slice.
			The difference images are amplified by a factor of five.
			The color bar of the difference images is at the right of the figure.}
	\label{fig:ablation_ssos}
\end{figure*}

\subsection{Radiological Evaluation}
Additionally, an experienced neuro-radiologist (L.S.) evaluated the reconstruction results of brain multi-coil data. 
Examples of the reconstruction results and ground truth images of HCP data are shown in Fig. \ref{fig:radiological_evaluation}.
As the arrowheads in the ground truth images of the first row indicate, the posterolateral margin of both putamina is clearly delineated, and it shows a certain contrast difference from the normal hypointensity area of the external capsule. 
However, in the supervised method, the posterolateral margin of putamina is not clearly defined. 
On the other hand, the proposed method successfully delineates the contrast difference between the posterolateral margin and the hypointensity area of the external capsule.
Furthermore, in the ground truth images in the second row, multiple T2 hyperintensity foci by perivasular space is observed in the subcortical white matter of bilateral cerebral hemispheres (the arrowheads). 
However, supervised learning shows only small number of perivascular space at the left frontal area. 
Meanwhile, our method reconstructs more perivascular space than supervised method. 
In general, through the radiological evaluation, we verify that there are no significant differences between supervised learning and proposed method in terms of clinical findings, even providing more clinical information in some cases as shown in Fig. \ref{fig:radiological_evaluation}.

\section{Discussion}\label{sec:discussion}
In our experiments, the conventional cycleGAN\cite{zhu2017unpaired} showed inferior qualitative and quantitative results. 
Since the conventional cycleGAN requires two generators and two discriminators, significantly large number of parameters have to be optimized. 
This leads to an unstable training due to the many local minimizers and the network is prone to overfitting error. 
Moreover, in conventional cycleGAN, each generator has to fool its matching discriminator, while reconstructing the input generated from another generator simultaneously.
Accordingly, if one generator or discriminator does not work properly, it can affect the performance of the other networks. 
On the other hand, when one generator is replaced by deterministic operation as in the proposed OT-cycleGAN architecture, the number of unknown parameters decreases, and it can improve the stability of the network training and lead to significant performance gain compared to conventional cycleGAN. 

Another interesting observation is that the reconstruction results of our unpaired learning method are comparable to the supervised method. 
This may be because the proposed cycleGAN learns the transportation mapping by minimizing the transportation cost between all combinations of samples in the two domains, while the supervised method learns the point-to-point mapping. 
Therefore, this unpaired training appears to increase the generalization power of the network, although more theoretical analysis needs to be performed to verify the empirical observations.

Recall that our PLS cost in \eqref{eq:costMR} is calculated in the SSoS image domain rather than complex image domain, even though our generator $G_\Theta$ generates multi-channel complex images.
One could argue that the cost should be also calculated in the complex image domain.
Although this was successful for the single coil cases\cite{sim2019optimal}, we found that such approach failed to provide meaningful results for the multi-coil cases, and we suspect that discriminating real and fake images for each coil image and combining them turns out to be a difficult task.

To verify our claim, we have conducted additional ablation study in  Fig. \ref{fig:ablation_ssos}.
More specifically, Fig. \ref{fig:ablation_ssos}(a) show  the reconstruction results when both of $\ell_{cycle}$ and $\ell_{WGAN}$ are calculated without SSoS.
In these cases, the generator produces blurry outputs with remaining aliasing artifacts.
We conjecture that this is because the generator and the discriminator are trained to learn coil sensitivity maps more than underlying images.
Fig. \ref{fig:ablation_ssos}(b) show that the quality of reconstructed images significantly decreases when SSoS is not used only in $\ell_{cycle}$.
Performance drop is more severe compared to the case in Fig. \ref{fig:ablation_ssos}(a), where SSoS is not used for both  $\ell_{cycle}$ and $\ell_{WGAN}$.
This is because the discriminator disturbs the reconstruction of coil sensitivity maps, since $\ell_{WGAN}$ is calculated with SSoS and the discriminator does not care about sensitivity maps.
On the other hand, $\ell_{cycle}$ computes the error for each coil.
We believe that this inconsistency introduces distortion that leads to the higher reconstruction error.
In Fig. \ref{fig:ablation_ssos}(c), only $\ell_{WGAN}$ is calculated without SSoS, which shows improved results, compared to Fig. \ref{fig:ablation_ssos}(a) and (b), but outputs are still blurred and some diagonal artifacts appear in the reconstructed images.
This is because the generator should not only reconstruct SSoS images to minimize $\ell_{cycle}$, but also consider the coil sensitivity maps for reducing $\ell_{WGAN}$.

Finally, in Fig. \ref{fig:ablation_ssos}(d), by defining the PLS cost using SSoS images, the design of the discriminator and the cycle-consistency term becomes simpler, and we can obtain accurate reconstruction.
As briefly discussed before, incorporating the nonlinearity in the forward problem is not a problem in our cycleGAN formulation, although it causes the difficulty in the standard PLS methods due to the difficulty of calculating the gradient.
This again shows the flexibility of our cycleGAN formulation.

Having said this, the phase information can be very important for some clinical applications.
In these applications, unsupervised learning for high performance reconstruction for individual coil image beyond the SSoS images may be necessary, which is an interesting future research direction.

\section{Conclusion}\label{sec:conclusion}
In this paper, we proposed a novel unpaired deep learning method for accelerated MRI using a modified cycleGAN architecture. 
It was shown that the proposed cycleGAN architecture can be derived from the optimal transport theory with a specially designed penalized least squares cost. 
Also, we verified that our method outperforms the conventional cycleGAN and provides competitive results compared to the supervised method. 
We believe that our method can be a powerful framework for MRI reconstruction by providing means to overcome the difficulty in acquiring matched reference data. 

\section{Acknowledgement}
This work is supported by National Research Foundation (NRF) of Korea, Grant number NRF-2020R1A2B5B03001980.

\bibliographystyle{IEEEtran}
\bibliography{ref,biblio_book}

\end{document}